\documentclass[twocolumn]{IEEEtran}
\usepackage{amsmath,amssymb,url}
\newtheorem{thm}{Theorem}
\newtheorem{proof}{Proof}
\newtheorem{lem}{Lemma}
\newtheorem{defn}{Definition}

\newtheorem{cor}{Corollary}

\newcommand{\Z}{\mathbb{Z}}

\newcommand{\e}{\varepsilon}

\newcommand{\LA}{\mathbf{LA}_n}
\newcommand{\DR}{\mathbf{DR}_n}
\newcommand{\TJC}{\mathbf{TJC}_n}

\def\a{\alpha}

\def\d{\delta}
\def\e{\varepsilon}

\begin{document}
\title{Orthogonal Designs and a Cubic Binary Function}

\author{Sophie Morier-Genoud
\hskip1cm Valentin Ovsienko
\thanks{S. Morier-Genoud,
Universit\'e Paris 6,
IMJ, 
UFR 929 de Math\'ematiques,
Case 247, 4~pl. Jussieu,
75005 Paris, France;
sophiemg@math.jussieu.fr
} \thanks{V. Ovsienko,
CNRS,
ICJ,
Universit\'e  Lyon~1,
43 bd. du 11 novembre 1918,
69622 Villeurbanne cedex,
France;
ovsienko@math.univ-lyon1.fr
}}

\maketitle

\begin{abstract}
Orthogonal designs are fundamental mathematical 
notions used in the construction of space time block codes for wireless transmissions. 
Designs have two important parameters, the rate and the decoding delay;
the main problem of the theory is to construct designs maximizing the rate
and minimizing the decoding delay.

All known constructions of CODs are inductive or algorithmic.
In this paper, we present an explicit construction of optimal CODs.
We do not apply recurrent procedures and do calculate the matrix elements directly.
Our formula is based on a cubic function in two binary $n$-vectors. 
 In our previous work (Comm. Math. Phys., 2010, and J. Pure and Appl. Algebra, 2011),
 we used this function to define a series of non-associative algebras generalizing the classical
 algebra of octonions and to obtain sum of squares identities of Hurwitz-Radon type.
\end{abstract}
\begin{IEEEkeywords}
Orthogonal designs, decoding delay, 
maximal rate, peak-to-average power ratio, space-time codes, generalized octonions.
\end{IEEEkeywords}

\section{Introduction}
\label{sec1}

Orthogonal designs first appeared in the classical work of Hurwitz \cite{Hur1}, \cite{Hur} 
and Radon \cite{Rad},
in order to solve the problem of sum of squares identities (also known as composition of quadratic forms).
This problem can be formulated in different ways and related to many mathematical questions 
(normed division algebras, vector fields on spheres, Clifford modules, 
immersion of projective spaces in euclidean spaces...) arising in different fields. 
The general problem is widely open and keep inspiring work of many mathematicians, 
see \cite{Squa} and~\cite{Sha} for surveys.
In the 1970's,  orthogonal designs and their generalizations have been extensively studied 
from combinatorial and number theoretic viewpoints, see Geramita et al. 
\cite{GGW}--\cite{GeGe} and references
therein.

Orthogonal designs keep attracting much attention, since
they are used to construct space-time block codes 
for wireless communication with multiple transmit antennas. 
This idea was introduced by Tarokh, Jafarkhani and Calderbank \cite{TJC}, as a
generalization of the Alamouti scheme \cite{Ala} for wireless communication with two antennas.
Space-time block codes built out of the orthogonal designs achieve full transmit diversity 
and have a simple maximum likelihood decoding algorithm.

In this paper, we describe a method of construction of orthogonal designs.
Unlike all known constructions which are inductive,
i.e., use block matrices of small sizes to construct bigger matrices,
our construction calculates elements of the matrices directly.
In particular, we construct designs
satisfying optimal criteria of \cite{Lia} and \cite{AKP}, \cite{AKM}.
We also construct designs of type \cite{TJC} and \cite{DaRa2}
defined by matrices that have no zero elements.

\subsection{Definitions and known results}

\begin{defn}
A \textit{real orthogonal design} (ROD) of type $[p,n,k]$ is a matrix $G$ of size 
$p\times n$
with real entries $0,\pm x_1,\cdots,\pm x_{k}$, 
satisfying
$$
G^T\,G=\left(x_1^2+\cdots+x_k^2\right)I_n,
$$
where $G^T$ is the transpose matrix of $G$.
\end{defn}

\begin{defn}
A \textit{complex orthogonal design} (COD) of parameters $[p,n,k]$ is a matrix $H$ of size 
$p\times n$
with complex entries $0,\pm z_1,\cdots,\pm z_{k}$, and  their conjugates 
$\pm z_1^*,\cdots,\pm z_{k}^*$,
 satisfying
$$
H^*H= \left(|z_1|^2+\cdots+|z_k|^2 \right)I_n,
$$
where $H^*$ is the complex conjugate transpose of $H$.
\end{defn}

\begin{defn}
Given a $[p,n,k]$-(R or C)OD, the ratio $\frac{k}{p}$ is called the \textit{rate} 
of the design and the parameter $p$ is called the \textit{decoding delay} of the design.
\end{defn}

The main problem in the construction of real or complex orthogonal designs is 
to maximize the rate $\frac{k}{p}$  and minimize the delay $p$ for a given $n$.
The following answers have been provided
\begin{enumerate}
\item 
 $[p,n,k]$-ROD of rate $1$ exist for all $n$, and in this case the minimum delay is
$p=2^{\delta(n)}$, where
$$
\begin{array}{l}
 \delta(n)=
    \begin{cases}
\frac{n}{2}      & \text{ if $n=2,4,6\mod8$ }, \\[4pt]
\frac{n-1}{2}   & \text{ if $n=1, 7\mod8$ }, \\[4pt]
\frac{n+1}{2}   & \text{ if $n=3, 5\mod8$ },\\[4pt]
\frac{n}{2}-1    & \text{ if $n=0\mod8$}.
  \end{cases}
\end{array}
$$
This is a way to formulate the classical theorem or Hurwitz and Radon.
\item 
Using a doubling process of ROD of rate 1, Tarokh et al. \cite{TJC} obtain 
COD of rate $\frac{1}{2}$ and decoding delay $2^{\delta(n)+1}$;\\
We will denote by $\TJC$ this class of CODs, the parameters are
$$
\left[2^{\delta(n)+1},n,2^{\delta(n)}\right].
$$
\item
 Liang \cite{Lia} proves that the maximal rate of a $[p,n,k]$-COD with $n\not=p$ is 
$\frac{1}{2}+\frac{1}{n}$, if $n$ is even, and $\frac{1}{2}+\frac{1}{n+1}$, if $n$ is odd.
\item 
Adams et al. \cite{AKP} and \cite{AKM} find a tight lower bound for
 the decoding delay $p$ in a non-square COD 
achieving the maximal rate given by a binomial coefficient.
Let $n=2m-1$ or $n=2m$, then 
$$
p\geq{{2m}\choose{m-1}},
$$ 
for $n=0,1,3\mod4$
and 
$$
p\geq2{{2m}\choose{m-1}},
$$ 
for $n=2\mod4$.
We will denote by $\LA$ the class of CODs achieving the maximal rate and minimal decoding delay.
\item 
Liang \cite{Lia} and Lu et al. \cite{Lu} give
algorithms to produce designs of type $\LA$.
Another construction is given in~\cite{DaRa}.

\item 
Das and Rajan \cite{DaRa2} construct CODs of rate $\frac{1}{2}$ and decoding delay $2^{\delta(n)}$.
We will denote by $\DR$ this class of CODs, the parameters are
$$
\left[2^{\delta(n)},\,n,\,2^{\delta(n)-1}\right].
$$
It is interesting that the decoding delay of these CODs is twice lower than that of $\TJC$.
\end{enumerate}

Let us stress that CODs of rate $\frac{1}{2}$ are of interest, 
since for large values of $n$ the maximal rate is almost $\frac{1}{2}$. 
For instance, for $n=12$ the COD $\LA$ has rate $\frac{7}{12}$ and
decoding delay $792$, whereas $\DR$ has rate $\frac{1}{2}$
and decoding delay $64$, see \cite{DaRa2} for a comparative table.

\subsection{Main results and organization of the paper}

The goal of this article is to present a unified construction of RODs and CODs.

We start with an explicit formulas for the matrices $G$ of size $2^r\times2^r$
satisfying the conditions of ROD. 
The rows and columns of the matrices are labelled by elements in the set 
$\Z_2^r$, i.e., of $r$-vectors with coefficients 0 and 1.
The entries in the matrices are given by an explicit function 
$f:\Z_2^r\times \Z_2^r\rightarrow \Z_2$.
We then explain an easy way to reduce a RODs to a CODs.

With our method we  construct RODs with parameters:
\begin{enumerate}
\item[\textbullet] $\left[2^{\delta(2r)},\,2r,\,2^{\delta(2r)}\right]$,

that will produce CODs with same parameters as
$\TJC$ (using a doubling process) and with same parameters as
$\DR$ (using a reduction process), provided $n\not=1 \mod 8$;
\medskip
\item[\textbullet] 
$\left[2{{r+1}\choose{(m-1}},\, 2r,\, 2{{r}\choose{m}}\right]$, if $r=0,1,3\mod4$,

\medskip
$\left[4{{r+1}\choose{(m-1}},\, 2r,\, 2{{r}\choose{m}}\right]$, if $r=0\mod4$,
\medskip

where $m$ is defined by $r=2m-1$ or $2m$.
These RODs will reduce to COD with same parameters as $\LA$.
\end{enumerate}

The paper is organized as follows.
The next section contains the main ingredients of our approach.
We construct  RODs
with parameters that are twice the parameters of the optimal CODs.

Section \ref{RtoC} describes the reduction from the
$[p,n,k]$-RODs to $\left[\frac{p}{2},\,\frac{n}{2},\,\frac{k}{2}\right]$-CODs,
leading to optimal CODs of type $\LA$.

Section \ref{IndSec} presents a procedure that
allows us to construct a $\left[\frac{p}{2},\,n,\,k \right]$-COD out of a $[p,n,k]$-ROD,
provided the ROD is stable under the duality.
We thus obtain the CODs of types $\TJC$ and $\DR$.
Let us mention that the corresponding matrices have no zero entries.

Proofs of technical statements, as well as properties
of the binary functions we use, are collected in Appendix.

\section{General construction of RODs}

\subsection{Combinatorics over $\Z_2$}
We denote by $\Z_2^r$
the set of $r$-vectors $u=(u_1,\ldots,u_r)$,
where $u_i= 0$ or $1$\footnote{We use the notation $\Z_2=\{0,1\}$ 
for the abelian group of rank 2, other
notations: $\mathbb{F}_2$ and $\Z/2\Z$ are also often used.}.
The \textit{Hamming weight} $|u|$ of an element is the number of non-zero component, i.e.
$$
|u|=\#\left\{u_i=1\right\}_{1\leq i \leq r}.
$$
The sum of two elements $u$ and $v$ is just the sum componentwise modulo 2.
Every element is a sum of the basis vectors
$$
\e^j=(0,\ldots,0,1,0, \ldots,0),
$$
with 1 at $j$-th position.
We will also consider the element of maximal weight $r$:
$$
\e=
(1,1,\ldots,1).
$$
We will use the involution on $\Z_2^r$, that we call the ``hat duality'':
\begin{equation}
\label{Dua}
\hat{u}=u+\e^1,
\end{equation}
i.e., the change of 1st coordinate.

The following function in two arguments
$f:\Z_2^r\times \Z_2^r\to\Z_2$ 
plays the key r\^ole in our approach:
$$
\boxed{
f(u,v)=\sum_{i<j<k}\left(
u_iu_jv_k 
+u_iv_ju_k 
+v_iu_ju_k 
\right)
+ \sum_{i\leq j}\;u_iv_j,
}
$$
We will also use the function in one variable
$\a(u):=f(u,u)$,
given explicitly by
$$
\alpha(u)=\sum_{i<j<k}\;u_iu_ju_k + \sum_{i\leq j}\;u_iu_j.
$$
The value of $\alpha(u)$ depends only on the weight of $u$:
$$
\alpha(u)=\left\{ 
\begin{array}{ll}
0 & \text{ if } |u|=0 \mod 4,\\
1, &\text{ otherwise. }
\end{array}
\right.
$$
The function $f$ is used in all the constructions to determine signs, while $\a$
is used as a ``statistic'' to select good elements of~$\Z_2^r$.
Properties of $f$ and $\a$ are presented in Appendix.
 
\subsection{General construction of RODs}
In this section, we construct $(p\times{}n)$-matrices whose rows and columns are indexed
by subsets $W\subset \Z_2^r$ and $V\subset \Z_2^r$ of cardinality $p$ and $n$, 
respectively.

We define the matrix $G_u$, $u\in\Z_2^r$ by
$$
G_u=
\begin{array}{cc}
{}_v&\\[4pt]
\left(
\begin{array}{cccc}
&\vdots&\\[4pt]
&\vdots&\\[6pt]
\cdots&G_u^{w,v}&\cdots\\
&\vdots&\\
\end{array}
\right)
&\left.
\begin{array}{l}
\\
\\
\\[6pt]
\!\!\!\!\!{}_w\\
\\
\end{array}\right.
\end{array}
$$
where the entry in position $(w,v)$ is 
$$
G_u^{w,v}=\left\{ 
\begin{array}{ll}
(-1)^{f(u,v)}, & \text{ if } u=v+w,\\[4pt]
0, &\text{ otherwise. }
\end{array}
\right.
$$

The following properties are obvious:
\begin{enumerate}
\item
the matrix $G_u$ has at most one nonzero element on each row and on each column;
\item
if $V=W=\Z_2^r$, then $G_u$ is a 
$(2^r\times2^r)$-matrix with exactly one non-zero element  on each row and column.
\end{enumerate}

\begin{defn}
We call $U,V,W$ an \textit{admissible triple} if the following two conditions
are satisfied:
\begin{enumerate}
\item
$$
W=U+V,
$$ 
i.e., $u+v\in{}W$ for all $u\in{}U$ and $v\in{}V$ and
every element $w\in{}W$ can be written in the form $w=u+v$.
\item
If a non-zero element $w\in{}W$ decomposes in two ways:
$w=u+v=u'+v'$ then
\begin{equation}
\label{Cond}
\alpha(u+u')=\alpha(v+v')=1.
\end{equation}
\end{enumerate}
\end{defn}

\begin{thm}
\label{ThmOne}
{\it If $U,V,W$ is an admissible triple, then one has:

(i)
$
G_u^T\,G_u=I_n,
$
for all $u\in{}U$.

(ii)
$
G_u^T\,G_{u'}+G_{u'}^T\,G_u=0,
$
for all $u\not=u'\in{}U$.}\\
\end{thm}

This theorem is proved in \cite{MGO} and \cite{LMO}.
For the sake of completeness, we include the proof into Appendix.\\

\begin{cor}\label{consROD}
 If $U,V,W$ is an admissible triple, then
 
(i)
the matrix 
$$G=\sum_{u\in U}\;x_uG_u$$
 is a ROD with parameters 
$\left[\#W,\#V,\#U\right]$ in real variables $x_u$, $u\in U$,
where $\#$ is the cardinality of a set;

(ii)
in the case of rate 1, i.e., where $k=p$,
the matrix $G$ has \textit{no zero entries}.\\
\end{cor}

Our next task is to construct admissible triples $U,V,W$.
We will use the fact that the function $\alpha$ vanishes only on the elements whose weight is a multiple of 4.
Note that the easiest way to guarantee condition (\ref{Cond})
 is to choose the set $V$ so that $\a(v+v')=1$
for all $v,v'\in{}V$.

\subsection{RODs of rate 1}\label{RODTJC}
In this section, we provide triples of sets $U,V,W$ that produce real orthogonal designs
$$
G=\sum_{u\in U}\;x_uG_u
$$
of rate 1, with minimum delay, i.e. the parameters of $G$ are 
$\left[2^{\d(n)}, n, 2^{\d(n)}\right]$,
(provided $n\not=1 \mod 8$).\\

\noindent
\textbf{Case} $r=0,1,2 \mod 4$.
One chooses the following subsets
$$
\begin{array}{lcl}
V&=& \left\{ \e^j, \widehat{\e^j}, \,1\leq j \leq r \right\},\\[6pt]
U&=& W\;=\;\Z_2^r,
\end{array}
$$
where $\hat{}$ is the duality (\ref{Dua}).
Then $G$ is a $\left[2^r, 2r, 2^r\right]$-ROD.
\footnote{
This ROD is optimal,
except for the case $r=0\mod 4$, where,
according to the Hurwitz-Radon theorem, there
is a $\left[2^r, 2r+1, 2^r\right]$-ROD. 
We do not dwell here on a more involved construction to produce such a ROD.}\\

\noindent
\textbf{Case} $r=3 \mod 4$.
One chooses the following subsets
$$
\begin{array}{lcl}
V&=&\left\{\e,\;  \widehat{ \e}, \; \e^j,\;  \widehat{\e^j}, \; 1\leq j \leq r\right\},\\[6pt]
U&=& W\;=\;\Z_2^r,
\end{array}
$$
the $G$ is a $\left[2^r, 2r+2, 2^r\right]$-ROD.

\subsection{Non-square RODs of rate $\frac{1}{2}+\frac{1}{2m}$}\label{RODLA}
All the RODs below have maximal rate $\frac{1}{2}+\frac{1}{2m}$,
 when $r=2m$ or $2m-1$. \\

\noindent
\textbf{Case} $r=1,2 \mod 4$.
Consider $r=2m-1$ or $r=2m$ and choose the set $U$ of the elements
of weight $m$ and their dual,
the set $V$ is chosen as in the first case:
$$
\begin{array}{lcl}
U&=& \left\{u, \hat u: |u|=m \right\},\\[6pt]
V&=& \left\{ \e^j, \; \widehat{\e^j},\;  1\leq j \leq r \right\},\\
\end{array}
$$
It follows that the space $W=U+V$ is:
$$
\begin{array}{lcl}
W&=& \left\{u: |u|=
m-1,\, m,\,m+1
 \right\}\;\bigcup\\[6pt]
&& \left\{u: u_1=1\; ,\; |u|=m+2 \right\}\;\bigcup
\\[6pt]&&
\left\{u: u_1=0 \;,\; |u|=m-2 \right\}.
\end{array}
$$
The matrix $G$ is a ROD with parameters
$$
\textstyle
\left[2{r+1\choose m-1}, \,2r,\, 2{r\choose m}\right],\quad
\left[4{r\choose m-1},\,2r,\, 2{r\choose m} \right],
$$
for odd $r$ and even $r$, respectively.\\

\noindent
\textbf{Case} $r=0 \mod 4$.
Consider $r=2m$ (where $m$ is even) and choose the following subsets
$$
\begin{array}{lcl}
U&=& \left\{u: u_1=1\; ,\; |u|=m \right\} \bigcup \\[4pt]
&&\; \left\{u: u_1=0 \;,\; |u|=m-1 \right\},\\[10pt]
V&=& \left\{\e,\,\hat{\e},\, \e^j, \,\widehat{\e^j},\;  2\leq j \leq n \right\},\\[10pt]
W&=&  \left\{u: u_1=1\; ,\; |u|=m-1,m+1 \right\} \bigcup\\[4pt]
&&\; \left\{u: u_1=0 \;,\; |u|=m-2,m \right\},\\
\end{array}
$$
then $G$ is a $\left[2{r\choose m-1},\,2r,\,  2{r-1\choose m-1}\right]$-ROD.\\

\noindent
\textbf{Case} $r=3 \mod 4$.
Consider $r':=r+1$ and apply the previous case with $r'=0 \mod 4$
to obtain a 
$
\left[2{r+1\choose m-1},\,2r+2,\, 2{r\choose m} \right]
$-ROD,
where $r=2m-1$.
Removing two columns, we obtain a 
$\left[2{r+1\choose m-1}, \,2r,\, 2{r\choose m}\right]$-ROD.\\

\noindent
In each of the above cases, condition (\ref{Cond})
is satisfied for all $v,v'\in{}V$.\\

To finish this section, let us mention that the binary numeration
have already been efficiently used in \cite{DaRa}, \cite{DaRa2}
to construct RODs and CODs of maximal rate.
In particular, subsets of $\Z_2^r$ similar to our sets $U,V$ and $W$
were described.
The main difference of our approach is the function $f$
and explicit construction of the matrices.

\section{Reduction from ROD to COD}\label{RtoC}

In this section, we present a procedure to reduce a $\left[ 2p,2n,2k \right]$-ROD
to a $\left[p,n,k \right]$-COD.
Such a procedure is not always possible, it requires nice properties of sets $U,V,W$.
We first describe the general procedure of reduction
and then apply it to RODs of rate 1 constructed in Section \ref{RODTJC}.

\subsection{The general procedure}
The main idea is to use a duality
$$
\widehat{}\,:\, \Z_2^r\to\Z_2^r
$$ 
defined by
 $\hat u=u+e,$
 where  the element $e\in\Z_2^r$ satisfies
$
f(e,e)=1,
$
and to choose sets $U,V$ and $W$ stable under 
the duality:
 $$
 \widehat{U}=U, \quad  \widehat{V}=V, \quad  \widehat{W}=W,
 $$
In practice, we use the hat duality (\ref{Dua}), i.e., $e=\e^1$.

Given a ROD of type $\left[ 2p,2n,2k \right]$ defined by sets $U,V$ and $W$ in $\Z_2^r$,
our goal is to reduce it to a COD with parameters $\left[p,n,k \right]$.
 The method consists in two steps. 
 First, we introduce splitting of the sets $U,V,W$,
 in order to decompose the matrices $G=\sum_{u\in U}x_uG_u$ into admissible $(2\times2)$-blocks. 
Second, we replace the admissible blocks
by complex variables, $z_u=x_u+ix_{\hat u}$ or $z^*_u=x_u-ix_{\hat u}$.\\

STEP 1:
 We fix
 the following splitting of $U$
\begin{equation}
\label{defu0}
\begin{array}{rcl}
 U_0&:=&\left\{u\in U: f(u,e)=0\right\},\\[6pt]
 U_1&:=& \left\{u\in U: f(u,e)=1 \right\}.
\end{array}
\end{equation}
Note that $\widehat{U_0}=U_1$, cf. Property (b) of $f$
in Appendix.

We now need to find subsets $V_0$, $V_1$, $W_0$ and $W_1$
satisfying the following conditions
\begin{equation}\label{splitvw}
\begin{array}{rcccl}
V&=&V_0\bigsqcup V_1,&& V_1=\widehat{V_0};\\
W&=&W_0\bigsqcup W_1, && W_1=\widehat{W_0};\\
W_0&=&U_0+V_0&=&U_1+V_1;\\
W_1&=&U_0+V_1&=&U_1+V_0,
\end{array}
\end{equation}
where $\bigsqcup$ denotes the disjoint union.

These splitting induce a natural decomposition of the matrices~$G_u$
into $(2\times 2)$-blocks whose columns are labelled by $(v,\hat v)\in V_0\times V_1$ and 
the rows by $(w,\hat w)\in W_0\times W_1$ :
$$
G_u=
\begin{array}{cc}
v\quad\hat v&\\
\left(
\begin{array}{cccc}
&\vdots \quad\vdots&\\[5pt]
&\vdots\quad\vdots&\\[8pt]
\cdots&&\cdots\\[-12pt]
&\boxed{\widetilde{G}_u^{w,v}}&\\[-10pt]
\cdots&&\cdots\\
&\vdots\quad\vdots&\\
\end{array}
\right)
&\left.
\begin{array}{l}
\\
\\
\\
\\
\!\!\!\!w\\
\!\!\!\!\hat w\\
\\
\end{array}\right.
\end{array}
$$
where $u=v+w$ and so $\hat u= v+ \hat w= \hat v+ w$.

For $u\in U_0$ (and therefore $\hat u\in U_1$), the non-zero blocks are of the form 
$$
\widetilde{G}^{w,v}_u=\left(
\begin{array}{cc}
(-1)^{f(u,v)}&0\\[6pt]
0&(-1)^{f(u,\hat v)}
\end{array}
\right)
$$
and 
$$
\widetilde{G}^{w,v}_{\hat{u}}=\left(
\begin{array}{cc}
0& (-1)^{f(\hat u,v)}\\[6pt]
(-1)^{f(\hat u,\hat v)}&0
\end{array}
\right);
$$
non-zero blocks are located at the same place in $G_u$ and $G_{\hat u}$.

Moreover, since $f$ is linear in the 2nd variable,
$$
\begin{array}{lclcl}
f(u,\hat v)&=&f(u,v)+f(u,e)&=&f(u,v),\\[4pt]
f(\hat u,\hat v)&=&f(\hat u,v)+f(\hat u,e)&=&f(\hat u,v)+1,
\end{array}
$$
so that the entries in the blocks of $G_u$ are of the same sign and those of $G_{\hat u }$ are 
of the opposite sign.\\

STEP 2:
The matrix $G=\sum_{u\in U}x_uG_u$ decomposes into $(2\times 2)$-blocks, and the non-zero blocks 
are of two types
\begin{equation*}
(T1) \quad
\pm
\left(
\begin{array}{rr}
x_u& x_{\hat u}\\[4pt]
-x_{\hat u}&x_u
\end{array}
\right),
 \text{ or }
 (T2)\quad
\pm
\left(
\begin{array}{rr}
x_u& -x_{\hat u}\\[4pt]
x_{\hat u}&x_u\\
\end{array}
\right).
\end{equation*}
We construct a complex matrix $H$ from $G$ by substituting to the block (T1) 
the complex variable $\pm z_u$ and to the block $(T2)$ the complex conjugate variable $\pm z_{u}^*$.
More precisely, the entry of $H$ in position $(w,v)\in W_0\times V_0$ is
$$
H^{w,v}=
\left\{
\begin{array}{l}
(-1)^{f(v+w,v)}z_{v+w},\, \text{if } f(\widehat{v+w},v)=f(v+w,v),\\[4pt]
(-1)^{f(v+w,v)}z^*_{v+w},\, \text{otherwise,}
\end{array}
\right.
$$
if $v+w\in U_0$, and $H^{w,v}=0$, otherwise.\\

\begin{thm}\label{consCOD}
The constructed matrix 
$H$
defines a COD with parameters 
$\left[p,\,n,\,k\right]$.
\end{thm}

\subsection{CODs of parameters $\LA$}

As application of the above procedure, let us reduce
the RODs constructed in Section \ref{RODLA}, in order to obtain
the optimal CODs  of type $\LA$.
We need to describe here the subsets 
$U_0,U_1, V_0,V_1,W_0,W_1$ 
satisfying \eqref{defu0} and \eqref{splitvw}.

From the expression of $f$ we see that 
$f(u,\e^1)$ depends only on the class $|u|\mod4$.
More precisely
$$
\begin{array}{c|c|c|c|c}
|u| \mod4&\quad0\quad&\quad1\quad&\quad2\quad&\quad3\quad\\[6pt] \hline \hline
f(u,\e^1) \text{ if } u_1=0&0&0&1&1\\[4pt] \hline
f(u,\e^1) \text{ if } u_1=1&0&1&1&0\\[4pt]
\end{array}
$$
for an arbitrary $u\in\Z_2^r$.\\

\noindent
\textbf{Case} $r=1,2 \mod 4$.
Let now $u\in{}U$, so that $|u|=m$ and
$r=2m-1$ or $2m$.
In this case, $m$ is necessarily odd.
\begin{enumerate}
\item[\textbullet] if $m=1\mod4$, then for $u\in U$ we have
$$
f(u,\e^1)=0 \Longleftrightarrow u_1=0,
$$
in other words,
$$
U_0=\left\{u\in{}U\,|\,u_1=0\right\}, \quad U_1=\left\{u\in{}U\,|\,u_1=1\right\}.
$$
We easily check that
$$
\begin{array}{ll}
V_0=\left\{v\in{}V\,|\,v_1=0\right\}, & V_1=\left\{v\in{}V\,|\,v_1=1\right\},\\[4pt]
 W_0=\left\{w\in{}W\,|\,w_1=0\right\},& W_1=\left\{w\in{}W\,|\,w_1=1\right\}.
 \end{array}
$$
satisfies property \eqref{splitvw}.

\item[\textbullet] if $m=3\mod4$, then for $u\in U$ we have
$$
f(u,\e^1)=0 \Longleftrightarrow u_1=1,
$$
in other words,
$$
U_0=\left\{u\in{}U\,|\,u_1=1\right\}, \quad U_1=\left\{u\in{}U\,|\,u_1=0\right\}.
$$
We easily check that
$$
\begin{array}{ll}
V_0=\left\{v\in{}V\,|\,v_1=0\right\}, & V_1=\left\{v\in{}V\,|\,v_1=1\right\},\\[4pt]
 W_0=\left\{w\in{}W\,|\,w_1=1\right\}& W_1=\left\{w\in{}W\,|\,w_1=0\right\}.
 \end{array}
$$
again satisfies property \eqref{splitvw}.\\
\end{enumerate}

\noindent
\textbf{Case} $r=0 \mod 4$.
In this case, $m$ is defined by $r=2m$,
so that $m$ is even.
\begin{enumerate}
\item[\textbullet] 
If $m=0\mod4$, then for $u\in U$ we have
$$
f(u,\e^1)=0 \Longleftrightarrow u_1=1,
$$
in other words,
$$
U_0=\left\{u\in{}U\,|\,u_1=1\right\}, \quad U_1=\left\{u\in{}U\,|\,u_1=0\right\}.
$$
We easily check that
$$
\begin{array}{ll}
V_0=\left\{v\in{}V\,|\,v_1=0\right\}, & V_1=\left\{v\in{}V\,|\,v_1=1\right\},\\[4pt]
 W_0=\left\{w\in{}W\,|\,w_1=1\right\},& W_1=\left\{w\in{}W\,|\,w_1=0\right\}.
 \end{array}
$$
have the desired property.
\item[\textbullet] 
If $m=2\mod4$, then for $u\in U$ we have
$$
f(u,\e^1)=0 \Longleftrightarrow u_1=0,
$$
in other words,
$$
U_0=\left\{u\in{}U\,|\,u_1=0\right\}, \quad U_1=\left\{u\in{}U\,|\,u_1=1\right\}.
$$
We easily check that
$$
\begin{array}{ll}
V_0=\left\{v\in{}V\,|\,v_1=0\right\}, & V_1=\left\{v\in{}V\,|\,v_1=1\right\},\\[4pt]
 W_0=\left\{w\in{}W\,|\,w_1=0\right\}& W_1=\left\{w\in{}W\,|\,w_1=1\right\}.
 \end{array}
$$
have the desired property.
\end{enumerate}

\section{Induction}\label{IndSec}

The second procedure that we call
\textit{induction} allows one to transform a $\left[p,n,k\right]$-ROD to 
a $\left[p,n,\frac{k}{2} \right]$-COD. 
This induction consists in complexification composed with reduction;
it can be applied whenever the set $U$, $V$ and $W$ are stable under 
an involution $\tilde u=u+e$ for an element $e$ 
satisfying $f(e,e)=1$.

\subsection{The general construction}
STEP 1:
We consider the splitting of $U$ as in (\ref{defu0}) and 
define the following subsets of 
$\Z_2^{r+1}=\Z_2^r\times 0\,\sqcup\,\Z_2^r\times 1$
\begin{eqnarray*}
 U'&=&U_0\times 0\sqcup U_1\times 1,\\[4pt]
V'&=&V\times 0 \sqcup V\times 1,\\[4pt]
W'&=&W\times 0 \sqcup W\times 1.
\end{eqnarray*}
For each $u\in U$, we embed the previous matrix $G_u$ into a twice bigger matrix,
 $G_{(u,\tau)}$, where $(u,\tau)\in U'$, with columns indexed by $W'$
and rows indexed by $V'$.
The matrix $G_{(u,\tau)}$ is composed by $(2\times2)$-blocks:
$$
G_{(u,\tau)}=
\begin{array}{cc}
{}_{(v,0)(\hat v,1)}&\\
\left(
\begin{array}{cccc}
&\vdots \quad\vdots&\\[4pt]
&\vdots\quad\vdots&\\[6pt]
\cdots&&\cdots\\[-12pt]
&\boxed{\widetilde{G}_{u}^{w,v}}&\\[-10pt]
\cdots&&\cdots\\
&\vdots\quad\vdots&\\
\end{array}
\right)
&\left.
\begin{array}{l}
\\
\\
\\[4pt]
\!\!\!\!_{(w,0)}\\
\!\!\!\!_{(\hat w,1)}\\
\\
\end{array}\right.
\end{array},
$$
where the non zero blocks correspond to $(u,\tau)=(v+w,\tau)$ 
and coincide with those of $G_u$.
These data provide a $[2p,2n,k]$-ROD.\\

STEP 2: For each $u\in U_0$, we define
$$
G'_u=x_uG_{(u,0)}+x_{\hat u}G_{(\hat u,1)}.
$$
These matrices decompose into blocks that are all of type (T1) or (T2).
We then apply the reduction procedure to obtain a $\left[p,n,\frac{k}{2}\right]$-ROD.

\subsection{CODs with parameters $\TJC$ and $\DR$}
Consider the RODs of rate 1 constructed in Section \ref{RODTJC}.\\

Following \cite{TJC}, one can apply the following obvious doubling process:
$$
G^{(2)}:=\left(
\begin{array}{l}
G\\
G'
\end{array}
\right)
$$
where the variables $x_u$ in $G$ are now considered as complex variables,
and $G'$ is copy of $G$ associated to the conjugate variables, i.e. is defined by 
$$
G'=\sum_{u\in U}\;x^*_uG_u.
$$
The matrix $G^{(2)}$ is a COD of rate $\frac{1}{2}$, with parameters 
$$
\left[2^{\d(2r)+1},\,2r,\, 2^{\d(2r)}\right],
$$
that are precisely the parameters of $\mathbf{TJC}_{n}$,
for even $n$.
Removing a column, we are led to CODs with parameters $\mathbf{TJC}_{n}$,
for odd $n$, except for $n=1\mod 8$.\\

Applying the induction procedure to the RODs of rate~1, leads to
CODs with parameters 
$$
\left[2^{\d(2r)},\,2r,\, 2^{\d(2r)-1}\right],
$$
that are precisely the parameters of $\mathbf{DR}_{n}$,
for even $n$.
Again, removing a column, we obtain CODs of type $\mathbf{DR}_{n}$,
for odd~$n$, except for $n=1\mod 8$.\\

In both cases, the obtained CODs have no zero entries.\\

The missing case $n=1\mod 8$ escapes from the technique 
used in this paper.
This is due to the fact that we do not obtain a $\left[2^r, 2r+1, 2^r\right]$-ROD 
with $n=0\mod4$.

\section{Appendices}

\subsection{Properties of the functions $f$ and $\a$.}

The function $f$ has quite remarkable properties that we
briefly discuss here.

It is impossible to reconstruct the function $f$
from the function in one variable $\a$
(which is nothing but the restriction of $f$ to the diagonal in $\Z_2^r\times\Z_2^r$).
However, $\a$ contains the essential characteristics of $f$,
such as its symmetrization.

\begin{enumerate}

\item\label{pol} 
First polarization formula:
$$
f(u,v)+f(v,u)= \a(u+v)+\a(u)+\a(v).
$$
\item\label{polBis} 
Second polarization formula:
$$
\begin{array}{l}
f(u,v)+f(u,v+w)+f(u+v,w)+f(v,w)=\\[4pt]
\hskip2cm \a(u+v+w)\\[2pt]
\hskip2cm +\a(u+v)+\a(u+w)+\a(v+w)\\[2pt]
\hskip2cm +\a(u)+\a(v)+\a(w).
\end{array}
$$
\end{enumerate}
\noindent
These properties can be checked directly.
Note that
the expression in the right-hand-side of \ref{pol}) is called the \textit{coboundary} of $\a$,
it has a deep cohomological meaning.
The expression in the left-hand-side of \ref{polBis}) 
 is the coboundary of $f$, its measures the non-associativity of a certain algebra
 defined by $f$, see \cite{MGO}.
 Finally,  the expression in the right-hand-side of \ref{polBis})  is called the \textit{polarization}
of the cubic form $\a$.
Let us mention that, unlike the theory of quadratic forms,
the theory of cubic forms is not well developed in characteristic 2,
 not much is known.

Let us also give here more elementary properties of $f$
already used in the above constructions:

\begin{itemize}
\item[(a)] Linearity of $f$ in 2nd variable:
$$f(u,v+v')=f(u,v)+f(u,v').$$
\noindent

\item[(b)] Pseudo-linearity in 1st variable:
$$
f(u+v,v)=f(u,v)+f(v,v).
$$
\end{itemize}

We invite the reader to consult \cite{MGO} for more
information about $f$ and $\a$.

\subsection{Proof of Theorem \ref{ThmOne}.}
We apply the formula for matrices multiplication.
The coefficient in position $(v,v')$ in the product $G^T_uG_{u'}$ is
$$
(G^T_uG_{u'})^{v,v'}=
\left\{
\begin{array}{ll}
(-1)^{f(u,v)+f(u',v')}& \text{if } v+v'=u+u',\\[4pt]
0 & \text{otherwise}.
\end{array}
\right.
$$
This implies that $G^T_uG_{u'}=I_n$, and $G^T_uG_{u'}+G^T_{u'}G_{u}=0$ if and only if
$$
(-1)^{f(u,v)+f(u',v')}+(-1)^{f(u',v)+f(u,v')}=0
$$
whenever $v+v'=u+u'$.
The above condition is equivalent to
$$
f(u,v)+f(u',v')+f(u',v)+f(u,v')=1.
$$

\begin{lem}
If $u+u'=v+v'$ then
$$
f(u,v)+f(u',v')+f(u',v)+f(u,v')=\a(u+u').
$$
\end{lem}

\begin{proof}
Rewrite the left-hand-side using $v'=u+u'+v$ and the
linearity in the 2nd variable,
after cancellation of double terms
one obtains
$$
f(u',u)+f(u',u')+f(u,u)+f(u,u').
$$
This reduces to $\a(u+u')$ using the first polarization formula.
\end{proof}

Theorem \ref{ThmOne} follows.

\subsection{Proof of Theorem \ref{consCOD}.}\label{PrSect}

First notice that in each column of the matrix $H$ the symbol $z_u$
appears exactly once (``symbol $z_u$'' means one of the following four elements:
$\pm z_u, \pm z^*_u$).
This implies that the diagonal entries in $H^*H$ are all equal to
$$
\sum_{u\in U_0}|z_u|^2.
$$
It remains to show that the non-diagonal entries in $H^*H$ are all zero.
We show that, in the hermitian product of two distinct columns of $H$, the terms
pairwise cancel.

Consider the four entries of $H$, in position $(w,v)$, $(w',v)$, $(w',v)$ and $(w',v')$.
$$
H=
\begin{array}{cc}
\begin{array}{cccccc}
\quad&\quad v&\quad&\; v'&\quad&\end{array}
&\\
\left(
\begin{array}{cccccc}
&\vdots&&\vdots&\\
\cdots&z_1&\cdots&z_2&\cdots\\
&\vdots&&\vdots&\\
\cdots&z_3&\cdots&z_4&\cdots\\
&\vdots&&\vdots&\\
\end{array}
\right)
&\left.
\begin{array}{l}
\\
\!\!\!\!w
\\
\\[8pt]
\!\!\!\!w'\\
\\
\end{array}\right.
\end{array}
$$
\textbf{Case I}: there exist $u,u'$ in $U_0$ such that
\begin{equation*}\label{4ent}
u=v+w=v'+w', \quad  u'=v+w'=v'+w.
\end{equation*}
In this case, the four entries are non zero and one has
$$
z_1,z_4\in 
\left\{\pm z_u, \; \pm z_u^* \right\}, \quad z_2,z_3\in 
\left\{\pm z_{u'}, \;\pm z_{u'}^*\right\}.
$$
The corresponding blocks in the matrix $G$
$$
G=
\left(
\begin{array}{cccccc}
&\vdots&&\vdots&\\
\cdots&A_1&\cdots&A_2&\cdots\\
&\vdots&&\vdots&\\
\cdots&A_3&\cdots&A_4&\cdots\\
&\vdots&&\vdots&\\
\end{array}
\right)
$$
come from $x_uG_u+x_{\hat u}G_{\hat u}+x_{u'}G_{u'}+x_{\hat {u'}}G_{\hat {u'}}$
and therefore satisfy
$$
A^T_1A_2+A^T_3A_4=0.
$$
This translates to
$$
z_1^*z_2+z_3^*z_4=0.
$$
\textbf{Case II}: there do not exist $u,u'$ in $U_0$ such that
\begin{equation*}
u=v+w=v'+w', \quad  u'=v+w'=v'+w.
\end{equation*}
In this case at least one of the following situations holds
$$
z_1=z_4=0 \quad \text{ or } \quad z_2=z_3=0,
$$
and, again, $z_1^*z_2+z_3^*z_4=0$.\\

We have proved that the columns of $H$ are pairwise orthogonal (with respect to the
hermitian product).
And we conclude finally that
$$
H^*H=\left(\Sigma_{u\in U_0}|z_u|^2\right)I_{n}
$$
where $n=\#V_0$.

\subsection{Orthogonal designs and Hurwitz problem of sums of squares}

It is well-known that the existence of $[p,n,k]$-ROD is related to the 
Hurwitz problem on composition of quadratic forms, \cite{Hur},\cite{Rad}
(see also \cite{Sha} for a survey).
\begin{defn}
A Hurwitz sum of squares identity (SSI) of size $[p,n,k]$ is an identity
\begin{equation}
\label{rsN}
\left(a_1^2+\cdots{}+a_k^2 \right) \left(b_1^2+\cdots{}+b_n^2 \right)
=c_1^2+\cdots{}+c_p^2,
\end{equation}
where $c_i$ are bilinear expressions in  $a_i$ and $b_i$ with
integral coefficients (the elements $a_i,b_i,c_i$'s are considered here as real variables).
Such an identity will be referred as a $[p,n,k]$-identity.
\end{defn}

It is known that if such an identity holds then the integral coefficients in the expressions of 
$c_i$'s can be chosen among $\{0, 1, -1\}$.
Hurwitz proved the following fundamental theorem. 
{\it There exists a $[p,n,k]$-ROD if and only if there exists a $[p,n,k]$-SSI.}

Let us recall here how the equivalence can be established.
Since $c$'s are linear in $a$'s and $b$'s, one has
\begin{equation}\label{sumsq}
c=\left(\sum_{1\leq i \leq k}a_iA_i\right)b,
\end{equation}
where $b$ is a column-vector with components $b_i$ and
$c$ is a column-vector with components $c_i$, and where
$A_i$ are $p\times n$ matrices (with entries $0,1,-1$)
One then easily checks that the identity (\ref{sumsq}) holds if and only if
\begin{equation}\label{Hurwmat}
A_i^T\;A_i=I_n, \quad A_i^TA_j+A_j^TA_i=0,\quad \forall\; i\not=j.
\end{equation}
Then, the matrix $A=\left(\sum_{1\leq i \leq k}a_iA_i\right)$ is a $[p,n,k]$-ROD in
the variables $a_i$'s. 

A classical result of Hurwitz \cite{Hur1} states that $[n,n,n]$-SSI 
exist if and only if $n=1,2,4,8$.
This statement relies on classification of normed division algebras,
see \cite{Set} for a survey on relations between division algebras and wireless communications.
The case of $[n,n,k]$-SSI was solved independently
by Hurwitz \cite{Hur} and Radon \cite{Rad},
this is the origin of the famous Hurwitz-Radon function.

Let us mention that our previous results \cite{MGO} and \cite{LMO} 
were formulated in terms of partial solutions to the Hurwitz problem.



\end{document}